\documentclass[12pt]{amsart}
\usepackage{amsmath,amscd,amsfonts,amssymb, color,bbm, bm}
\usepackage{latexsym, graphicx, pstricks,rotating, enumerate}
\usepackage[pdftex,bookmarks,colorlinks]{hyperref}
\definecolor{darkblue}{rgb}{0.0,0.0,0.3}
\hypersetup{colorlinks=false,citecolor=darkblue,
unicode=true,pdftitle=tomography.pdf,pdfauthor=Ritabrata
Sengupta,bookmarks=false} 
\numberwithin{equation}{section}
\newtheorem{prop}{Proposition}[section]
\newtheorem{defin}{Definition}[section]
\newtheorem{rem}{Remark}[section]

\newtheorem{theorem}{Theorem}[section]

\newtheorem{corollary}{Corollary}[section]

\newcommand{\tr}{\mathrm{Tr}}
\newcommand{\dd}{\mathrm{d}}
\newcommand{\mr}{\mathrm}
\newcommand{\mb}{\mathbf}
\newcommand{\re}{\mathrm{Re}}
\newcommand{\im}{\mathrm{Im}}

\newcommand{\ket}[1]{\ensuremath{\left|#1\right\rangle}}

\begin{document}
\title{From particle counting to Gaussian tomography}
\author{K. R. Parthasarathy}
\address{Theoretical Statistics and Mathematics Unit, Indian
Statistical Institute, Delhi Centre, 7 S J S Sansanwal Marg,
New Delhi 110 016, India}
\email[K R Parthasarathy]{\href{mailto:krp@isid.ac.in}{krp@isid.ac.in}}
\author{Ritabrata Sengupta}
\address{Theoretical Statistics and Mathematics Unit, Indian
Statistical Institute, Delhi Centre, 7 S J S Sansanwal Marg,
New Delhi 110 016, India}
\email[Ritabrata Sengupta]{\href{mailto:rb@isid.ac.in}{rb@isid.ac.in}}
\begin{abstract}

\par The momentum and position observables in an $n$-mode
boson Fock space $\Gamma(\mathbb{C}^n)$ have the whole real
line $\mathbb{R}$ as their spectrum. But the total number
operator $N$ has a discrete spectrum
$\mathbb{Z}_+=\{0,1,2,\cdots\}$. An $n$-mode Gaussian state
in $\Gamma(\mathbb{C}^n)$ is completely determined by the
mean values of momentum and position observables and their
covariance matrix which together constitute a family of
$n(2n+3)$ real parameters. Starting with $N$ and its unitary
conjugates by the Weyl displacement operators and operators
from a representation of the symplectic group $Sp(2n)$ in
$\Gamma(\mathbb{C}^n)$ we construct $n(2n+3)$ observables
with spectrum $\mathbb{Z}_+$ but whose expectation values in
a Gaussian state determine all its mean and covariance
parameters. Thus measurements of discrete-valued
observables enable the tomography of the underlying Gaussian
state and it can be done  by
using  5 one mode and 4 two mode Gaussian symplectic gates
in single and pair mode wires of $\Gamma(\mathbb{C}^n)
= \Gamma(\mathbb{C})^{\otimes n}$. Thus the tomography
protocol admits a simple description in a language  similar
to circuits in quantum computation theory. Such a Gaussian
tomography applied to outputs of a Gaussian channel  with
coherent input states permit a tomography of the channel
parameters. However, in our procedure the number of counting
measurements exceeds the number of channel parameters
slightly. Presently, it is not clear whether a more
efficient method exists for reducing this tomographic
complexity. 

\par As a byproduct of our approach an elementary derivation
of the probability generating function of $N$ in a Gaussian
state is given. In many cases the distribution turns out to
be infinitely divisible and its underlying L\'evy measure
can be obtained. However, we are unable to derive the exact
distribution in all cases. Whether this property of infinite
divisibility holds in general is left as an open problem.
\end{abstract}
\thanks{RS acknowledges financial support from the National
Board for Higher Mathematics, Govt. of India.} 
\maketitle
\section{Introduction}
\par It is in the nature of quantum theory that properties of
the state of quantum systems can be inferred from
measurements of observables of physical significance taking
values in a discrete set or a continuum. From an
experimental point of view it is natural to seek as much
information as possible from the discrete measurements. A
typical measurement of the discrete type is counting the
number of particles of a particular type. Suppose the
unknown state of a system can be described in terms of some
parameters which constitute a manifold of dimension $k$.
Then it is natural to look for $k$ discrete-valued observables
from whose expectation values one can determine the values
of these parameters. Of course, such observables should be
physically meaningful and also experimentally measurable.
This has been extensively studied in the book
\cite{MR2181445}.

\par In this article we explore this problem of determining
the state when it is known that the state is Gaussian.
Consider the Hilbert space $L^2(\mathbb{R}^n)$, or
equivalently, the boson Fock space $\Gamma(\mathbb{C}^n)$
over the $n$-dimensional complex Hilbert space
$\mathbb{C}^n$. A Gaussian state of $n$-modes in
$L^2(\mathbb{R}^n)$ is completely described by its momentum
and position mean values and a $2n \times 2n$ covariance
matrix. Thus, an $n$-mode Gaussian state is determined by
$n(2n+3)$ parameters. Here one has observables $a_j^\dag
a_j$, the number of particles in the $j$-th mode for each
$j=1,2,\cdots,n$ and also their unitary equivalents in
different frames which are obtained by Weyl (displacement)
operators as well as unitary operators which implement the
symplectic linear transformations in the position and momentum
observables obeying the canonical commutation relations.
Using these resources we shall construct $n(2n+3)$ number
observables which have the discrete spectrum $\{0, 1, 2,
\cdots \}$ and the property that all the means and
covariances of the unknown Gaussian state can be easily
determined from their expectation values. Measurements on
these number observables on an ensemble of such a Gaussian
state and the law of large numbers can be used to estimate
the unknown parameters. In the process of such an
investigation we shall determine the probability generating
function of the distribution of the total number
observable $\sum_{j=1}^n a_j^\dag a_j$, its mean and
variance in any fixed Gaussian state. 

\par Following Heinosaari et al \cite{MR2683408} a 
Gaussian channel with $n$ degrees of
freedom is determined by a pair $(A,\,B)$ where $A$ is a $2n
\times 2n$ real matrix, $B$ is a $2n \times 2n$ real
positive semidefinite matrix satisfying the matrix
inequality 
\[ B + \imath (A^T J_{2n} A -J_{2n}) \geq 0,\]
where $J_{2n}$ is defined in equation (\ref{eq:2.12}). Thus such a
channel is determined by $4n^2 + n(2n+1)$ real
parameters. Such a channel yields an output Gaussian
state for any given input Gaussian state. By choosing a few
appropriate input coherent states and performing a Gaussian
tomography on the output Gaussian states, we show how the
matrices $A$ and $B$ of the quasifree channel can be
estimated. We use coherent states as they constitute an important
class of mathematical objects, which are easy to realize
experimentally. Further, creation of different coherent
states with different amplitudes and phases can be possible
from the same experimental set-up. We note that, a similar
approach of using coherent states for tomography (although
in a different scheme) has also been taken by Lobino et al
\cite{Lobino24102008}, and was further developed in
\cite{saleh, PhysRevA.88.022101}. 
\par Bosonic particle counting is an important tool in
quantum optics, both for theory and for experiments
\cite{Sc-W1, Sc-W2}. Gaussian states, channels and their
applications have been used extensively in quantum
information theory. These concepts have been studied in detail
in the book of Holevo \cite{zbMATH06069722} and in the survey
article by Weedbrook et al \cite{RevModPhys.84.621}. We also
refer to the books by de Gosson \cite{gosson}  and
Parthasarathy \cite{krp4} for the connection
between symplectic geometry and quantum stochastic calculus. We
refer to the survey article by Lvovsky and Raymer
\cite{RevModPhys.81.299} and the references therein, for
experimental processes of continuous variable tomography.

\par We organise the paper as follows. To increase the
readability of the article, in 
\S \ref{sec:prelim}, we write a short introduction to
notions like exponential vector, Weyl operator, Gaussian state,
Fourier transform, Gaussian channels and other necessary concepts which
will be used in subsequent sections. In this venture, we
mostly follow the approach taken in the papers
\cite{arvindsurvey, MR2662722, 2014arXiv1408.5686P}. In \S
\ref{sec:partstat} we derive the basic formulae for
expectation and variance of the total number operator which will be used in
\S \ref{sec:state}  to estimate an unknown Gaussian state.
In \S \ref{sec:channel} we use the tomographic method
derived for Gaussian states in \S \ref{sec:state} to
estimate the unknown parameters of a Gaussian channel.

\section{Notation and preliminaries} \label{sec:prelim}
\par In this section we give a short survey on Gaussian
state and other necessary concepts. Though these 
have been extensively studied in various references, we
follow the method and notation adopted in the book
\cite{krp4} and the papers \cite{MR2662722, zbMATH06185273,
2014arXiv1405.6476P, 2014arXiv1408.5686P}.
\subsection{Exponential vector}
\par Let $\mathcal{H}$ be a finite dimensional complex
Hilbert space. When $\dim \mathcal{H} =n$ and $\mathcal{H}$
is identified with $\mathbb{C}^n$ we express its elements as
column vectors $\mb{z}=(z_1, z_2, \cdots , z_n)^T$ with
$z_j$ being complex scalars and the scalar product between
two elements $\mb{z}$ and $\mb{z'}$ as 
\[\langle \mb{z}| \mb{z'} \rangle = \sum_{j=1}^n \bar{z_j}
z_j'.\]

\par Define the \emph{boson Fock space }
$\Gamma(\mathcal{H})$ over $\mathcal{H}$ by 
\begin{equation}\label{2.1}
\Gamma(\mathcal{H}) = \mathbb{C}  \oplus \mathcal{H} \oplus
{\mathcal{H}^\circledS}^2 \oplus \cdots \oplus
{\mathcal{H}^\circledS}^r \oplus \cdots 
\end{equation}
where $\circledS^r$ denotes $r$-fold symmetric tensor
product. Elements of the subspace
${\mathcal{H}^\circledS}^r$ in $\Gamma(\mathcal{H})$ are
called $r$-particle vectors and elements of the form 
\[u_0 \oplus u_1 \oplus \cdots \oplus u_r \oplus \cdots\]
where all but a finite number of $u_r$'s are null, are
called \emph {finite particle vectors}.  Finite particle
vectors constitute a dense linear manifold $\mathcal{F}$ in
$\Gamma(\mathcal{H})$. For any $\mb{u} \in \mathcal{H}$ we
associate the \emph{exponential vector} $\mb{e(u)}$ in
$\Gamma(\mathcal{H})$ defined by 
\begin{equation}
\mb{e(u)} = \mb{1} \oplus \mb{u} \oplus
\frac{\mb{u}^{\otimes 2}}{\sqrt{2!}} \oplus \cdots \oplus
\frac{\mb{u}^{\otimes r}}{\sqrt{r!}} \oplus \cdots.
\end{equation}
Then 
\begin{equation}
\langle \mb{e(u)} | \mb{e(v)} \rangle = \exp \langle \mb{u}
| \mb{v} \rangle \quad \forall \mb{u},\mb{v} \in
\mathcal{H}.
\end{equation}
The linear manifold $\mathcal{E}$ generated by all the
exponential vectors is called \emph{exponential domain}. The
two dense linear manifolds $\mathcal{E}$ and $\mathcal{F}$
are useful domains for constructing several operators of
physical significance.

\subsection{Weyl operator}
\par For any $\mb{u} \in \mathcal{H}$ we associate the
\emph{Weyl displacement operator} $W(\mb{u})$ by putting
\begin{equation} 
W(\mb{u}) \mb{e(v)} =e^{-\frac{1}{2} \|u\|^2 - \langle
\mb{u}| \mb{v}\rangle } \mb{e(u+v)}
\end{equation}
for all $\mb{v} \in \mathcal{H}$, observing that $W(\mb{u})$
is scalar product preserving on $\mathcal{E}$ and therefore
extends naturally to
$\Gamma(\mathcal{H})$. The Weyl operators obey the
multiplication property
\begin{equation}\label{eq:2.5}
W(\mb{u})W(\mb{v})=e^{-\imath \im \langle \mb{u}|
\mb{v}\rangle}W(\mb{u}+ \mb{v})
\end{equation}
for all $\mb{u},\,\mb{v}\in \mathcal{H}$ and yield a strongly
continuous, irreducible, factorizable and projective unitary
representation of the additive group $\mathcal{H}$. By
`factorizable' we mean the property that under the
isomorphism between $\Gamma(\mathcal{H}_1 \oplus
\mathcal{H}_2)$ and $\Gamma(\mathcal{H}_1) \otimes \Gamma(
\mathcal{H}_2)$ through the identification $\mathcal{I}
\mb{e}(\mb{u_1} \oplus \mb{u_2}) =\mb{e(u_1)} \otimes
\mb{e(u_2)}$ one has 
\[ \mathcal{I} W(\mb{u_1} \oplus \mb{u_2}) \mathcal{I}^{-1} = W(\mb{u_1})
\otimes W(\mb{u_2}).\]

\par If $\mb{u} \rightarrow \tilde{W}(\mb{u})$ is another
strongly continuous map from $\mathcal{H}$ into the unitary
group of a Hilbert space $\mathcal{K}$ such that equation
(\ref{eq:2.5}) holds with $W$ replaced by $\tilde{W}$ and
$\tilde{W}(\cdot)$ is irreducible, then
there exists a unitary isomorphism $V:\Gamma(\mathcal{H})
\rightarrow \mathcal{K}$ such that 
\[V W(\mb{u}) V^{-1} = \tilde{W}(\mb{u})\quad \forall
\mb{u}\in \mathcal{H}.\]
Thus the Weyl operators constitute a unique multiplicative
family up to unitary equivalence but with the presence of
the factor $\exp -\imath \im \langle \mb{u}| \mb{v} \rangle
$ in equation (\ref{eq:2.5}). We call $\mb{u} \rightarrow
W(\mb{u})$ the \emph{Weyl representation} of $\mathcal{H}$
and we shall exploit its properties to define a natural
quantum Fourier transform for states in
$\Gamma(\mathcal{H})$. For now, we shall introduce some
basic observations arising from the Weyl operators.

\par For any fixed $\mb{u}\in \mathcal{H}$, the map $ t
\mapsto W(t\mb{u})$ yields a strongly continuous one
parameter unitary group as $t$ varies in $\mathbb{R}$. By
Stone's theorem \cite{krp4} there exists a self adjoint
operator $p(\mb{u})$ such that 
\begin{equation}\label{eq:2.6}
W(t\mb{u}) =e^{-\imath t p(\mb{u})}, \quad t \in \mathbb{R}.
\end{equation}
Define the operators
\begin{eqnarray*}
q(\mb{u}) &=& -p(\imath \mb{u}),\\
a(\mb{u}) &=& \frac{1}{2} (q(\mb{u})+ \imath p(\mb{u})),\\
a^\dag(\mb{u}) &=& \frac{1}{2} (q(\mb{u}) - \imath
p(\mb{u})).
\end{eqnarray*}
All these operators have domains including the exponential
domain $\mathcal{E}$ and the domain $\mathcal{F}$ of finite
particle vectors. Indeed, any finite linear combination of
these operators have the same property and we denote their
respective closures by the same symbols. With this
convention one has 
\[W(\mb{u}) = e^{-\imath p(\mb{u})} = e^{a^\dag(\mb{u})
-a(\mb{u})}.\]
When $\mathcal{H} = \mathbb{C}^n = \mathbb{R}^n + \imath
\mathbb{R}^n$ and $\mb{u} = \mb{x} +\imath \mb{y}$ with
$\mb{x} =\re \mb{u}, ~ \mb{y} =\im\mb{u}$ we also have 
\[ W(\mb{u}) = W(\mb{x} + \imath \mb{y}) =e^{-\imath
(p(\mb{x}) - q(\mb{y}))}.\]
Furthermore, one has the following commutation relations.
\begin{eqnarray*}
\left[ p(\mathbf{u}), p(\mathbf{v}) \right] & = & 2\imath\,
\im \langle \mathbf{u} | \mb{v} \rangle, \\
\left[ q(\mathbf{u}), q(\mathbf{v}) \right] & = & 2\imath\,
\im \langle \mathbf{u} | \mb{v} \rangle, \\
\left[ q(\mathbf{u}), p(\mathbf{v}) \right] & = & 2\imath\,
\re \langle \mathbf{u} | \mb{v} \rangle, \\
\left[ a(\mathbf{u}), a(\mathbf{v}) \right] & = & 0,\\
\left[ a^\dag(\mathbf{u}), a^\dag(\mathbf{v}) \right] & = & 0,\\
\left[ a(\mathbf{u}), a^\dag(\mathbf{v}) \right] & = &
\langle \mb{u} | \mb{v} \rangle
\end{eqnarray*}
on the domains $\mathcal{E}$ and $\mathcal{F}$ for all
$\mb{u}$, $\mb{v}$ in $\mathcal{H}$.

\par Choose and fix an orthonormal basis
$\{\mb{e_j}\},~j=1,2,\cdots$ in $\mathcal{H}$, and define 
\begin{align*}
p_j &= \frac{1}{\sqrt{2}}p(\mb{e_j}), &
q_j&=-\frac{1}{\sqrt{2}}p(\imath \mb{e_j}),\\
a_j& =\frac{1}{\sqrt{2}}(q_j + \imath p_j), & a_j^\dag&
=\frac{1}{\sqrt{2}}(q_j - \imath p_j).
\end{align*}
Then one has the canonical commutation relations in the form 
\[ [p_r,p_s]=[q_r,q_s]=0,\quad [q_r, p_s]= \imath
\delta_{rs},\]
or equivalently,
\[ [a_r,a_s]=[a_r^\dag, a_s^\dag]=0,\quad [a_r,a_s^\dag] =
 \delta_{rs},\]
in the domains $\mathcal{E}$ and $\mathcal{F}$. The
observables $p_1,p_2,\cdots$ are called momentum
operators and $q_1,q_2, \cdots$ are called position
operators in the basis $\{ \mb{e_j}, \, j=1,2,\cdots\}$. 

\subsection{Quantum Fourier transform}
\par For any trace class operator $\rho$ in
$\Gamma(\mathcal{H})$ its \emph{quantum Fourier
transform} or simply \emph{Fourier transform}
$\hat{\rho}$ on $\mathcal{H}$ is defined by 
\[\hat{\rho}(\mb{u}) = \tr \rho W(\mb{u}),\quad \mb{u}\in
\mathcal{H}.\]
Then $\hat{\rho}$ is a  bounded continuous function of
$\mb{u}$ satisfying $\hat{\rho}(\mb{0}) = \tr \rho$. If
$\rho$ is positive then $\hat{\rho}$ obeys the Bochner
property: for any finite set $\{c_r,\, r=1,2,\cdots,k\}$ of
scalars and elements $\{\mb{u_r},\, r=1,2,\cdots,k\}$ in
$\mathcal{H}$ one has the inequality 
\[\sum_{r,s} \bar{c_r}c_s \exp(\imath \im \langle \mr{u_r} |
\mr{u_s} \rangle) \hat{\rho}(\mb{u_s}-\mb{u_r}) \geq 0.\]
Conversely, if $\varphi$ is a continuous function with
$\varphi(\mb{0})=1$ and $\varphi$ satisfies the Bochner property
above then there exists a unique state $\rho$ such that
$\hat{\rho} =\varphi$. When $\mathcal{H}=\mathbb{C}^n$ one has the
Fourier inversion formula:
\[ \rho =\frac{1}{\pi^n} \int \overline{\hat{\rho}(\mb{u})}
W(\mb{u}) \,\dd \mb{u}\]
where $\dd \mb{u}$ denotes integration with respect to the
$2n$ dimensional Lebesgue measure in $\mathbb{R}^{2n}$ with
$\dd \mb{u} = \dd\mb{x}\, \dd\mb{y}, ~\mb{u} = \mb{x} +
\imath \mb{y}, ~ \mb{x} = \re( \mb{u}),~ \mb{y} = \im (\mb{u})$.

\par With the help of Fourier transform we shall now
construct a natural Hilbert space isomorphism between the
Hilbert space of Hilbert-Schmidt operators in
$\Gamma(\mathcal{H})$ and the Hilbert space
$L^2(\mathbb{R}^{2n})$.

\begin{prop}\label{prop:2.1}
Let $\mathcal{H} = \mathbb{C}^n$. Then
\begin{equation*}
\frac{1}{\pi^n} \int \exp\left[ - \| \mb{w} \|^2 + \langle
\mb{u} | \mb{w} \rangle + \langle \mb{w} | \mb{v} \rangle
\right] \dd \mb{w} =\exp\langle \mb{u} | \mb{v} \rangle.
\end{equation*}
\end{prop}
\begin{proof}
Immediate from standard formulae for Gaussian integrals.
\end{proof}

\begin{prop}\label{prop:2.2}
Denote by $L^2(\mathcal{H})$, the Hilbert space of square
integrable functions on $\mathcal{H}$ with the scalar
product 
\[\langle f | g \rangle =  \int \overline{f(\mb{u})}
g(\mb{u}) \frac{\dd \mb{u}}{\pi^n}\]
and by $\mathcal{B}_2(\Gamma(\mathcal{H}))$ the Hilbert space
of all Hilbert-Schmidt operators on $\Gamma(\mathcal{H})$
with the scalar product 
\[\langle \rho_1 | \rho_2 \rangle = \tr \rho_1^\dag
\rho_2.\]
Then there exists a unique Hilbert space isomorphism
$\mathbb{F}: \mathcal{B}_2(\Gamma(\mathcal{H})) \rightarrow
L^2(\mathcal{H})$ such that for any $\mb{u},\, \mb{v}\in
\mathcal{H}$, 
\begin{equation}\label{eq:2.7}
\left(\mathbb{F} ( | \mb{e(u)} \rangle \langle \mb{e(v)} | )
\right)(\mb{w}) = \tr | \mb{e(u)} \rangle \langle \mb{e(v)} |
W(\mb{w})  \quad \forall \mb{w} \in\mathcal{H}.
\end{equation}
\end{prop}

\begin{proof}
The right hand side of equation (\ref{eq:2.7}) is equal to 
\begin{eqnarray} 
&& \langle \mb{e(v)} |W(\mb{w}) | \mb{e(u)} \rangle
\nonumber \\
&=&  e^{-\frac{1}{2} \|\mb{w} \|^2} \langle \mb{e(v)} |
e^{-\langle \mb{w} | \mb{u} \rangle} | \mb{e(u+w)} \rangle
\nonumber\\
&=& \exp \left[ -\frac{1}{2} \|\mb{w} \|^2 + \langle \mb{v}
| \mb{w} \rangle - \langle \mb{w} | \mb{u} \rangle + \langle
\mb{v}| \mb{u}\rangle \right]   \label{eq:2.8}
\end{eqnarray}
If we put $\rho_j = |\mb{e(u_j)}\rangle \langle \mb{e(v_j)}
|, ~ j=1,2$ then 
\begin{equation}\label{eq:2.9}
\tr \rho_1^\dag \rho_2 =\exp \left[ \langle \mb{u_1} |
\mb{u_2} \rangle + \langle \mb{v_2} | \mb{v_1} \rangle
\right].
\end{equation}
On the other hand the scalar product between $\tr |
\mb{e(u_j)} \rangle \langle \mb{e(v_j)} | W(\mb{w}) , ~ j=1,2$
reduces by Proposition \ref{prop:2.1} and equation
(\ref{eq:2.8}) to the right hand side of (\ref{eq:2.9}). Now
we observe that rank one operators of the form  $|\mb{e(u)}
\rangle \langle \mb{e(v)} |,~ \mb{u},\,\mb{v} \in
\mathcal{H}$ constitute a total set in
$\mathcal{B}_2(\Gamma(\mathcal{H}))$. Thus $\mathbb{F}$
defined by (\ref{eq:2.7}) extends uniquely to the whole of
$\mathcal{B}_2(\Gamma(\mathcal{H}))$. On the other hand
functions of $\mb{w}$ of the form on the right hand side of
(\ref{eq:2.8}) constitute a total set in $L^2(\mathcal{H})$.
Thus $\mathbb{F}$ extends to a Hilbert space isomorphism
between $\mathcal{B}_2(\Gamma(\mathcal{H}))$ and
$L^2(\mathcal{H})$.
\end{proof}

\subsection{Gaussian state}
\begin{defin}\label{def:2.3}
A state $\rho$ in $\Gamma(\mathcal{H})$ with $\mathcal{H}=
\mathbb{C}^n$ is called an $n$-mode Gaussian state if its
Fourier transform $\hat{\rho}$ is given by 
\begin{equation}\label{eq:2.10}
\hat{\rho}(\mb{x} + \imath \mathbf{y}) = \exp
\left[-\imath \sqrt{2} (\mathbf{l}^T\mathbf{x} -
\mathbf{m}^T\mathbf{y}) - \begin{pmatrix} \mathbf{x}\\
\mathbf{y} \end{pmatrix}^T S  \begin{pmatrix} \mathbf{x}\\
\mathbf{y} \end{pmatrix} \right].
\end{equation}
for all $\mb{x},~\mb{y}\in \mathbb{R}^n$ where $\mb{l},~
\mb{m}$ are elements of $\mathbb{R}^n$ and $S$ is a real $2n
\times 2n$ symmetric matrix satisfying the matrix inequality 
\begin{equation}\label{eq:2.11}
2S+\imath\, J_{2n} \geq 0
\end{equation}
with 
\begin{equation}\label{eq:2.12}
J_{2n} = \begin{bmatrix} 0 & - I_n \\ I_n & 0
\end{bmatrix},
\end{equation}
$I_n$ being the identity matrix of order $n$.
\end{defin}

\begin{rem}
\emph{ Equations (\ref{eq:2.10})--(\ref{eq:2.12}) have been written
keeping in mind the orders of the canonical momentum and
position observables as
$p_1,\,p_2,\,\cdots,\,p_n,\,q_1,\,q_2,\,\cdots,\,q_n$.
Sometimes it is more convenient to distinguish the different
modes of a Gaussian state by using the order
$p_1,q_1,\,p_2,q_2,\,\cdots,\,p_m,q_n$. This is usually
achieved by employing the permutation 
\[\sigma =\begin{pmatrix}
1 & 2 & 3 & 4& \cdots & 2n-1 & 2n\\
1 & n & 2 & n+1 & \cdots & n & n+n 
\end{pmatrix}.\]
Then the right hand sides of (\ref{eq:2.10})--(\ref{eq:2.12})
are obtained by changing $(\mb{x}^T,\mb{y}^T),~(\mb{l}^T,
\mb{m}^T)$ and $S$ respectively to $(x_1,y_1,x_2,y_2,\cdots,x_n,y_n)$, $
(l_1,m_1,l_2,m_2,\cdots,l_n,m_n)$, and $\sigma S
\sigma^{-1}$ with $J_{2n}$ replaced by 
\[\tilde{J_{2n}} = \begin{bmatrix}
\begin{array}{lr}
0 & -1 \\ 1 & 0
\end{array} &&& \\
& \begin{array}{lr}
0 & -1 \\ 1 & 0
\end{array} &&\\ 
 &  & \ddots & \\
&&& \begin{array}{lr}
0 & -1 \\ 1 & 0
\end{array}
\end{bmatrix},
\]
By abuse of notation we may denote both $J_{2n}$ and
$\tilde{J_{2n}}$ by the same symbol $J_{2n}$.}
\end{rem}

\par We choose the canonical orthonormal basis $\mb{e_j}
=(0,0,\cdots,1,0,\cdots,0)^T$ with $1$ in the $j$-th
position for $j=1,2,\cdots$, then the momentum and position
operators satisfy the relations
\[\tr \,p_j \rho =l_j, \quad \tr \, q_j \rho =m_j\]
and $S$ is the covariance matrix of $(p_1,p_2,\cdots,
p_n,-q_1, -q_2,\cdots , -q_n)$ in the state $\rho$
satisfying (\ref{eq:2.10}-\ref{eq:2.12}). Whenever
\ref{eq:2.10} is satisfied we write 
\[\rho = \rho_g(\mb{l},\mb{m};S).\]
Thus $\rho$ is completely described by $n(2n+3)$ parameters.

\begin{prop}\label{prop:2.4}
Let $\rho_g(\mb{l_j},\mb{m_j};S_j),~j=1,2$ be two $n$-mode
Gaussian states. Then 
\begin{equation}\label{eq:2.13}
\tr~ \rho_g(\mathbf{l_1},\mathbf{m_1};S)
\rho_g(\mathbf{l_2},\mathbf{m_2};T) = \frac{\exp\left[-\frac{1}{2}
\begin{bmatrix} \mathbf{l_1} - \mathbf{l_2}\\ -(\mathbf{m_1}
-\mathbf{m_2})\end{bmatrix}^T 
(S+T)^{-1} \begin{bmatrix} \mathbf{l_1} - \mathbf{l_2}\\ -(\mathbf{m_1}
-\mathbf{m_2})\end{bmatrix}\right]} { \sqrt{\det(S+T)}}.
\end{equation}
\end{prop}
\begin{proof}
Any state in $\Gamma(\mathbb{C}^n)$ is a positive operator
of unit trace and hence a Hilbert-Schmidt operator. Thus by
Proposition \ref{prop:2.2} we have
\begin{eqnarray*}
&&\tr~ \rho_g(\mathbf{l_1},\mathbf{m_1};S)
\rho_g(\mathbf{l_2},\mathbf{m_2};T) \\
&=& \frac{1}{\pi^n} \int \exp\left[ 
\imath \sqrt{2} ((\mathbf{l_1} - \mathbf{l_2})^T \mathbf{x}
-(\mathbf{m_1} -\mathbf{m_2})^T
\mathbf{y}) -  \begin{pmatrix} \mathbf{x}\\
\mathbf{y} \end{pmatrix}^T (S+T)  \begin{pmatrix} \mathbf{x}\\
\mathbf{y} \end{pmatrix} 
\right] \dd\mathbf{x} \, \dd\mathbf{y}.
\end{eqnarray*}
The rest follows from the standard formula for the
characteristic function of a multivariate normal density
function in statistics.
\end{proof}

\begin{prop}\label{prop:2.5}
For any $\mb{u} \in \mathbb{C}^n$ with $ \mb{x}
=\re(\mb{u}), ~\mb{y} =\im(\mb{u})$
\[ W(\mb{u}) \rho_g(\mb{l},\mb{m};S) W(\mb{u})^\dag =
\rho_g(\mb{l}',\mb{m}';S)\]
where 
\[\mb{l}' = \mb{l} + \sqrt{2} \mb{y}, \quad \mb{m}' = \mb{m}
+ \sqrt{2} \mb{x}.\]
\end{prop}

\begin{proof}
Immediate from Corollary 3.3  in \cite{MR2662722}. 
\end{proof}

\par We denote by $Sp(2n)$ the symplectic group
$Sp(2n,\mathbb{R})$ of all real $2n \times 2n$ matrices $L$
satisfying the relation 
\[L^T J_{2n} L = J_{2n}.\]
Let $\Gamma(L)$ be the unitary operator in
$\Gamma(\mathbb{C}^n)$ which is unique upto a scalar
multiple of modulus unity and satisfies the relation 
\[\Gamma(L) W(\mb{x} + \imath\, \mb{y}) \Gamma(L)^{-1} =
W(\mb{x}' + \imath\, \mb{y}'), \quad \forall \mb{x},\,
\mb{y} \in \mathbb{R}^n,\]
where 
\[L \begin{pmatrix} \mb{x} \\ \mb{y} \end{pmatrix} =
\begin{pmatrix} \mb{x}' \\ \mb{y}' \end{pmatrix}.\]

\par If $U$ is any $ n \times n$ unitary matrix and $U=A +
\imath \, B$ where $A=\re(U)$ and $B=\im(U)$ then the $2n
\times 2n$ matrix
\[L = \begin{bmatrix} A & -B \\ B & A \end{bmatrix}\]
is an orthogonal matrix which is also an element of
$Sp(2n)$. We denote the corresponding $\Gamma(L)$ by
$\Gamma(U)$ and call it the second quantization of $U$. We
can realize $\Gamma(U)$ as the unique unitary operator
satisfying 
\[\Gamma(U) \mb{e(u)} =\mb{e}(U\mb{u})\quad \forall \mb{u}
\in \mathbb{C}^n.\]
With these notations we have 
\begin{prop}\label{prop:2.6}
For any $L \in Sp(2n)$
\[\Gamma(L) \rho_g( \mb{l}, \mb{m}; S) \Gamma(L)^\dag =
\rho_g( \mb{l}', \mb{m}'; S')\]
where 
\begin{eqnarray*}
\begin{pmatrix} \mb{l}' \\ -\mb{m}' \end{pmatrix} & = &
(L^{-1})^T \begin{pmatrix} \mb{l}' \\ -\mb{m}'
\end{pmatrix}\\
S' &=& (L^{-1})^T S L^{-1}.
\end{eqnarray*}
\end{prop}
\begin{proof}
This is Corollary 3.5 of \cite{MR2662722}. 
\end{proof}

\section{Particle counts and their statistics in a Gaussian
state} \label{sec:partstat}
\par Let $\rho_g(\mathbf{l},\mathbf{m};S)$ be an $n$-mode
Gaussian state in $\Gamma(\mathbb{C}^n)$ and whose Fourier
transform is given by equation (\ref{eq:2.10}). Define
the observables 
\begin{eqnarray*}
N_j &= &a_j^\dag a_j = \frac{1}{2}(p_j^2 + q_j^2 -1), \quad 1
\leq j \leq n.\\
N &=& \sum_{j=1}^n N_j.
\end{eqnarray*}
Both $N_j$ and $N$ are observables with spectrum
$\{0,1,2,\cdots\}$. $N_j$ is called the number operator or
observable which counts the number of particles (photons) in
the $j$-th mode. Since the $N_j$'s commute with each other
they have a joint distribution in the state
$\rho_g(\mathbf{l},\mathbf{m};S)$ with support in
$\{0,1,2,\cdots\}^n$. Using Proposition \ref{prop:2.4} we
shall derive a formula for the probability generating
function of this joint distribution and arrive at some
natural corollaries. To this end we introduce the Gaussian
state 	

\[ \rho(\mathbf{0},\mathbf{0};T)   = \prod_{j=1}^n (1
- e^{-t_j}) e^{-\sum_{j=1}^n t_j a_j^\dag a_j};\]
where 
\begin{eqnarray*}
T &=&\begin{bmatrix}
\frac{1}{2}\left(\frac{1+e^{-t_1}}{1-e^{-t_1}}\right)I_2 & &\\
& \ddots&\\
& & \frac{1}{2}\left(\frac{1+e^{-t_n}}{1-e^{-t_n}}\right)I_2 
\end{bmatrix} \\
& = & D\left(
\frac{1}{2}\left(\frac{1+e^{-t_j}}{1-e^{-t_j}}\right)I_2,~
1\leq j \leq n \right);\qquad  t_j >0 \quad \forall j,
\end{eqnarray*}
with $D$ indicating the diagonal block matrix with blocks of order
$2\times 2$ and the diagonal entries being enumerated within
$(\quad)$. These are the well-known thermal states. It
follows immediately from Proposition \ref{prop:2.4},
equation(\ref{eq:2.13}) that
\[
\tr \rho_g(\mathbf{l},\mathbf{m};S)  e^{-\sum_{j=1}^nt_j a_j^\dag a_j} =\frac{\exp\left[ -\frac{1}{2} \begin{pmatrix}
\mathbf{l} \\ - \mathbf{m} \end{pmatrix}^T \left( S+D\left(
\frac{1}{2}\left(\frac{1+e^{-t_j}}{1-e^{-t_j}}\right)I_2,~
1\leq j \leq n \right)
\right)^{-1}  \begin{pmatrix} \mathbf{l} \\ - \mathbf{m}
\end{pmatrix} \right]}{\prod_{j=1}^n (1
-e^{-t_j}) \sqrt{\det  \left[S + D\left(
\frac{1}{2}\left(\frac{1+e^{-t_j}}{1-e^{-t_j}}\right)I_2,~
1\leq j \leq n \right) \right]} }.
\]
Substituting $N_j=a_j^\dag a_j$, $x_j=e^{-t_j}$ we get 
 \begin{equation} \label{eq:3.1}
\tr \rho_g(\mathbf{l},\mathbf{m};S) x_1^{N_1} \cdots
x_n^{N_n} = \frac{\exp\left[ -\frac{1}{2} \begin{pmatrix}
\mathbf{l} \\ - \mathbf{m} \end{pmatrix}^T \left( S+D\left(
\frac{1}{2}\left(\frac{1+x_j}{1-x_j}\right)I_2,~
1\leq j \leq n \right)
\right)^{-1}  \begin{pmatrix} \mathbf{l} \\ - \mathbf{m}
\end{pmatrix} \right]}{\prod_{j=1}^n (1
-x_j) \sqrt{\det  \left[S + D\left(
\frac{1}{2}\left(\frac{1+x_j}{1-x_j}\right)I_2,~
1\leq j \leq n \right) \right]} }
\end{equation}
for $0<x_j<1~ \forall j$. The right hand side of this
equation is nothing but the probability generating function
of the joint distribution of the observables $N_j$, $1\leq j
\leq n$.

\begin{theorem}\label{th:3.1}
Let $\rho_g(\mathbf{l},\mathbf{m};S)$ be an $n$-mode
Gaussian state in $\Gamma(\mathbb{C}^n)$ whose covariance
matrix $S$ has eigenvalues $\lambda_1 \ge \lambda_2 \ge
\cdots \ge \lambda_{2n}$ with respective real eigenvectors
$\mb{b_1}, \mb{b_2},\cdots, \mb{b_{2n}}$ and momentum
position mean $\begin{pmatrix} \mb{l} \\ \mb{m}
\end{pmatrix}$ satisfying 
\[\begin{pmatrix} \mb{l} \\ -\mb{m} \end{pmatrix} =
\sum_{j=1}^{2n} \tau_j \mb{b_j}.\]
Suppose 
\[ \alpha_j = \frac{\lambda_j -\frac{1}{2} } {\lambda_j+
\frac{1}{2}}, \qquad 1 \le j \le n.\]
Then the probability generating function $G_N(x)$ of the
distribution of the total number operator $N$ in the state
$\rho_g(\mathbf{l},\mathbf{m};S)$ is given by 
\begin{eqnarray}
G_N(x) & =& \tr ~ \rho_g(\mathbf{l},\mathbf{m};S) \, x^N
\nonumber \\
&  = & \prod_{j=1}^{2n}\sqrt{\frac{1 - \alpha_j}{1 -\alpha_j x}} \exp \left[
-\frac{1}{2} \tau_j^2 \frac{(1-x)(1-\alpha_j)}{(1-\alpha_j
x)} \right], \qquad 0 \le x <1.
\label{eq:3.2}
\end{eqnarray}

\end{theorem}

\begin{proof}
Putting $x_j=x$ for all $j$ in (\ref{eq:3.1}) and making use
of the eigenbasis and eigenvalues of $S$ we see that
(\ref{eq:3.1}) becomes 
\[G_N(x) =\prod_{j=1}^{2n} \frac{ \exp\left[ - \frac{1}{2}
\tau_j^2 \left( \lambda_j + \frac{1}{2}\frac{1+x}{1-x}
\right)^{-1} \right]} {\sqrt{(1-x) \lambda_j +\frac{1}{2}
(1+x)}}.\]
The  rest is elementary algebra using the definitions of
$\alpha_j$ in terms of $\lambda_j$ for every $j$.  
\end{proof}

\begin{corollary}
The probability distribution of $N$ in the state
$\rho_g(\mb{l}, \mb{m};S)$ satisfies the following:
\begin{eqnarray}
(i)&& \Pr (N=0) = \prod_{j=1}^{2n} (1 -\alpha_j) \exp\left[
-\frac{1}{2} \tau_j^2 (1 - \alpha_j) \right],\label{eq:3.3} \\
(ii) && \langle N \rangle = \frac{1}{2}\left[ \tr\,\left( S-
\frac{1}{2} \right) + \left\|  \begin{pmatrix} \mb{l}
\\ -\mb{m} \end{pmatrix} \right\|^2 \right], \label{eq:3.4}\\
(iii)&& \text{Variance}(N) = \frac{1}{2} \tr\,\left( S-
\frac{1}{2} \right) \left( S+ \frac{1}{2} \right) +
\begin{pmatrix} \mb{l} \\ -\mb{m} \end{pmatrix}^T S
\begin{pmatrix} \mb{l} \\ -\mb{m} \end{pmatrix},
\label{eq:3.5}
\end{eqnarray}
\end{corollary}
\begin{proof}
Property (i) follows by putting  $x=0$ in (\ref{eq:3.2}).
Properties (ii) and (iii) are obtained from (\ref{eq:3.2})
by taking the logarithm of $G_N$, differentiating twice and
taking $\lim_{x \rightarrow 1} \frac{\dd}{\dd x} (\log
G(x))$ and $\lim_{x \rightarrow 1} \frac{\dd^2}{\dd^2 x}
(\log G(x))$.
\end{proof}

\begin{rem}
\emph{ Equations (\ref{eq:3.3})-(\ref{eq:3.5}) show that the
probability of presence of a particle and the expectation
and variance of the total number of particles get enhanced
when the Gaussian state has a nonzero momentum position mean
vector. Indeed, the mean and variance of $N$ tend to infinity
as the length of the momentum position mean vector increases
to infinity. Equations (\ref{eq:3.3}) and (\ref{eq:3.5})
indicate the possibility of estimating the mean and
covariance parameters of a Gaussian state by measuring the
number operator under different displacements. We shall
discuss this approach to the tomography of a Gaussian state
in great detail in the next section \S\ref{sec:state}. }
\end{rem}

\par It may be noted that the parameters $\alpha_j$ in
Theorem \ref{th:3.1} satisfy the inequality $|\alpha_j|
<1$ for every $j$. If $\mb{l}$ and $\mb{m}$ are null-vectors
and $\alpha_j \ge 0$ for every $j$ then $G_N(x)$ assumes the
form 
\[G_N(x) =\prod_{j=1}^{2n} \left( \frac{1-\alpha_j}{1 -
\alpha_j x} \right)^{\frac{1}{2}}\]
and the corresponding distribution of $N$ is a convolution
of $2n$ negative binomial distributions of index
$\frac{1}{2}$, some of which may be degenerate at $0$. In
particular, it is infinitely divisible. In this case
\[\text{Variance} (N)  - \langle N \rangle = \frac{1}{2} \tr \left(
S -\frac{1}{2} \right)^2 \geq 0 \]
and the distribution exhibits a super Poissonian property.
When $S=\frac{1}{2} I_{2n}$ the state becomes the vacuum
state and $N$ has the degenerate distribution, degenerate at
$0$. 

\par Now we shall analyse the distribution of $N$ in a pure
Gaussian state. In this case $S=\frac{1}{2} L^T L $ for some
element $L$ of the group $Sp(2n)$ and its eigenvalues can be
expressed as 
\begin{equation}\label{eq:3.6}
\left( \frac{c_1}{2}, \frac{1}{2c_1}, \frac{c_2}{2},
\frac{1}{2c_2}, \cdots, \frac{c_k}{2}, \frac{1}{2c_k},
\frac{1}{2},\cdots, \frac{1}{2} \right)
\end{equation}
where $c_j>1$ for $1 \le j \le k$. Then 
\[
\begin{array}{l l ll l}
\alpha_{2j-1} &=& \frac{c_j -1}{c_j + 1}>0 & \text{ for } &
1 \leq j \leq k,\\
\alpha_{2j} &=& \frac{1- c_j }{1+c_j}<0 & \text{ for } &
1 \leq j \leq k,\\
\alpha_r &=& 0 & \text{ for } &
2k+1 \leq r \leq 2n.
\end{array}
\]
Write  $\beta_j = \alpha_{2j-1}$, $1 \leq j \leq k$ so that
$\alpha_{2j} = - \beta_j, ~ 1 \leq j \leq k$. Then the
probability generating function of the total number operator
$N$ in (\ref{eq:3.2}) assumes the form 
\begin{equation}\label{eq:3.7}
G_N(x) = G_1(x) G_2(x) G_3(x)
\end{equation}
where 
\begin{eqnarray}
G_1(x) & =& \prod_{j=1}^k
\sqrt{\frac{1-\beta_j^2}{1-\beta_j^2 x^2}} \label{eq:3.8}\\
G_2(x) & = & \exp \frac{1}{2} \sum_{j=1}^{k}
\left[\tau_{2j-1}^2  \frac{(x-1) (1-\beta_j)}{1-\beta_j x} +
\tau_{2j}^2  \frac{(x-1) (1+\beta_j)}{1 + \beta_j x}
\right] \label{eq:3.9}\\
G_3(x) & = & \exp\left[ \frac{1}{2} \sum_{j=2k+1}^{2n}
\tau_j^2 (x-1)\right]\label{eq:3.10}
\end{eqnarray}
where $0 < \beta_j <1$ for $1 \le j \le k$. 

\par Writing 
\begin{eqnarray}
\gamma_j & =& \tau_{2j-1}^2 (1 - \beta_j),
\label{eq:3.11}\\
\delta_j & =& \tau_{2j}^2 (1 + \beta_j), \quad 1 \le j \le
k, 
\label{eq:3.12}
\end{eqnarray}
one can express $G_2(x)$ in (\ref{eq:3.9}) as 
\begin{equation}\label{eq:3.13}
G_2(x)= \exp \frac{1}{2}\sum_{j=1}^k
 \frac{ \beta_j(\gamma_j -\delta_j)(x^2-1) + \left[
\gamma_j(1-\beta_j) +\delta_j(1+\beta_j)\right] (x-1)}{1 -
\beta_j^2 x^2} 
\end{equation}
where $0 < \beta_j < 1$, $ \gamma_j \ge 0, ~ \delta_j \ge 0$
for $1 \le j \le k$. 

\par From (\ref{eq:3.7}) it follows that $G_1(x)$ is the
probability generating function of a convolution of
probability distributions $\mu_j, ~ 1 \le j \le k$  where
the probability generating function of $\mu_j$ is equal to 
\[(1-\beta_j^2)^{\frac{1}{2}} \sum_{r=0}^\infty
\frac{1.3.5.\cdots (2r+1)}{r!} \beta_j^{2r} x^{2r}, \quad
1\le j \le k.\]
In particular $\mu_j$ is an infinitely divisible
distribution with support in $\{0,2,4,\cdots\}$. Equation
(\ref{eq:3.10}) shows that $G_3(x)$ is the probability
generating function of a Poisson distribution with mean
value $\frac{1}{2} \sum_{r=2k+1}^{2n} \tau_r^2$.

\par If $\gamma_j \ge \delta_j$ for every $j=1,2,\cdots, k$
then $G_2(x)$ is clearly the probability generating function
of an infinitely divisible distribution with support in
$\{0,1,2,\cdots\}$. Its L\'evy measure can be easily read
off from (\ref{eq:3.13}). Under this assumption, i.e.
$\gamma_j \geq \delta_j$ for each $1 \le j \le k$ it follows
that $N$ has an infinitely divisible distribution in the
pure Gaussian state we started with.

\par However, we do not know the answer to the question
whether the distribution of the total number operator $N$ in
every Gaussian state is infinitely divisible and hence of a
mixed Poisson type.

\section{From particle counting to the tomography  of a
Gaussian state}\label{sec:state}

\par A Gaussian state with $n$ modes can be constructed if
its momentum and position means $\mb{l}$, $\mb{m}$ and its
covariance matrix $S$ are known. Our aim is to express these
parameters in terms of the expectation values of conjugates
of the total number operator by a few elementary gates in
the Hilbert space $\Gamma(\mathbb{C}^n)$. To this end we shall make use of the
Weyl displacement operators $W(\mb{u}), ~u\in \mathbb{C}^n$
and the Gaussian symmetries $\Gamma(L), ~ L \in Sp(2n)$
described in Section \S\ref{sec:partstat}. We start with an
elementary but basic result for achieving this tomography of
a Gaussian state. For any observable $X$ denote by $\langle X
\rangle$ its mean value in the state relevant to the
context. 

\begin{theorem}\label{th:4.1}
Let $\rho_g( \mb{l}, \mb{m};S)$ be a Gaussian state and let
$N$ be the number operator in the $n$-mode Hilbert space
$\Gamma(\mathbb{C}^n)$. Then the following hold:
\begin{enumerate}[(i)]
\item For all $\mb{x},\, \mb{y} \in \mathbb{R}^n$
\begin{equation} \label{eq:4.1}
\left\langle W(\mb{x} + \imath \mb{y})^\dag N W(\mb{x} + \imath
\mb{y}) \right\rangle - \langle N \rangle  = \|\mb{x}\|^2 +
\|\mb{y}\|^2 + \sqrt{2}(\mb{y}^T \mb{l} + \mb{x}^T \mb{m})
\end{equation}

\item For any $L \in Sp(2n)$
\begin{equation}\label{eq:4.2}
\left \langle  \Gamma(L)^\dag N \Gamma(L) \right
\rangle -\langle N \rangle = \frac{1}{2} \left[ \tr\,
S\left(L^{-1}{L^{-1}}^T - I_{2n} \right) + \begin{pmatrix}
\mb{l} \\ - \mb{m} \end{pmatrix}^T \left(L^{-1}{L^{-1}}^T -
I_{2n} \right)  \begin{pmatrix} 
\mb{l} \\ - \mb{m} \end{pmatrix}\right].
\end{equation}
\end{enumerate}

\end{theorem}

\begin{proof}
We have from Proposition \ref{prop:2.5} and equation
(\ref{eq:3.4}) 
\begin{eqnarray*}
\left\langle W(\mb{x} + \imath\, \mb{y})^\dag N W(\mb{x} +
\imath\, \mb{y})\right\rangle & =& \tr ~
\rho_g(\mathbf{l},\mathbf{m};S) W(\mb{x} + \imath\, \mb{y})^\dag N
W(\mb{x} + \imath\, \mb{y})  \\
&=& \tr ~ W(\mb{x} + \imath\, \mb{y}) \rho_g(\mathbf{l},\mathbf{m};S)
W(\mb{x} + \imath\, \mb{y})^\dag N \\
& =& \frac{1}{2} \left[ \tr\left(S-\frac{1}{2}\right) +
\|\mb{l} + \sqrt{2} \mb{y}\|^2 +
\|\mb{m} + \sqrt{2} \mb{x}\|^2\right]. 
\end{eqnarray*}
Now subtracting the value of $\langle N  \rangle$ given by
(\ref{eq:3.4}) we obtain equation (\ref{eq:4.1}). 

\par Similarly , we have from Proposition \ref{prop:2.6} and
equation (\ref{eq:3.4})
\begin{eqnarray*}
\tr~\rho_g(\mathbf{l},\mathbf{m};S) \Gamma(L)^\dag N
\Gamma(L) &=& \tr~ \Gamma(L) \rho_g(\mathbf{l},\mathbf{m};S)
\Gamma(L)^\dag N \\
&=& \frac{1}{2} \left[ \tr~ \left( {L^{-1}}^T S L^{-1}
- \frac{1}{2} \right) + \begin{pmatrix} \mb{l} \\ -\mb{m}
\end{pmatrix}^T {L^{-1}}^T L^{-1} \begin{pmatrix} \mb{l} \\ -\mb{m}
\end{pmatrix}\right].
\end{eqnarray*}
Now, subtracting the value of $\langle N \rangle$ given by
(\ref{eq:3.4}) get get (\ref{eq:4.2}).
\end{proof}

\par For any unitary operator $U$ in the $1$-mode Hilbert
space $\Gamma(\mathbb{C})$ we say that the operator $I
\otimes \cdots \otimes I\otimes U \otimes I \otimes \cdots
\otimes I$ with $U$ in $j$-th position, acting in the
$n$-mode Hilbert space $\Gamma(\mathbb{C}^n) =
\underbrace{\Gamma(\mathbb{C}) \otimes \Gamma(\mathbb{C}) \otimes \cdots
\otimes\Gamma(\mathbb{C})}_{n \text{ times}}$  is the gate
$U$ applied on the $j$-th mode and denote it by $U^{(j)}$.
In Figure \ref{fig:4.1} we represent $U^{(j)}$ by
following the notion of circuit diagrams in quantum computation.
\begin{figure}[!h]
\begin{center}
\setlength{\unitlength}{2mm}
\begin{picture}(25,10)(-15,-5)
\put(-10,4.5){\line(1,0){20}}
\put(-10,3){\line(1,0){20}}
\put(-10,0){\line(1,0){14}}
\put(4,-1.5){\framebox(3,3){$U$}}
\put(7,0){\line(1,0){3}}
\put(-10,-4.5){\line(1,0){20}}
\put(0,0.5){$\vdots$}
\put(0,-3){$\vdots$}
\put(-20,0){$U^{(j)}=$}
\end{picture}
\end{center}
\caption{}\label{fig:4.1} 
\end{figure}
\newline Here each wire stands for a $\Gamma(\mathbb{C})$ and in
the $j$-th wire the unitary operator $U$ is applied.

\par Similarly, if $V$ is a unitary operator in the
$2$-mode Hilbert space $\Gamma(\mathbb{C}^2) =
\Gamma(\mathbb{C}) \otimes \Gamma(\mathbb{C})$ we construct
the unitary operator $V^{(i,j)}$ in the $2$ modes
representing the $i$-th and the $j$-th wire. For example
$V^{(1,2)}$ is represented in Figure \ref{fig:4.2}.
\begin{figure}[!h]
\begin{center}
\setlength{\unitlength}{2mm}
\begin{picture}(25,10)(-15,-5)
\put(-10,4){\line(1,0){14}}
\put(-10,2.5){\line(1,0){14}}
\put(4,1.5){\framebox(3.5,3.5){$V$}}
\put(7.5,4){\line(1,0){2.5}}
\put(7.5,2.5){\line(1,0){2.5}}
\put(-10,0){\line(1,0){20}}
\put(-10,-4.5){\line(1,0){20}}
\put(0,-3){$\vdots$}
\put(-20,0){$V^{(1,2)}=$}
\end{picture}
\end{center}
\caption{}\label{fig:4.2}
\end{figure}

\par When $i<j$ are not successive we can apply a
permutation to make them successive, apply $V$ and follow by
the reverse permutation. One may use Figure
\ref{fig:4.3}.
\begin{figure}[h]
\begin{center}
\setlength{\unitlength}{2mm}
\begin{picture}(38,17)(-18,-8)
\put(-15,8.5){\line(1,0){30}}
\put(-15,7){\line(1,0){30}}
\put(-15,3){\line(1,0){15}} \put(-18,2.5){$i$}
\put(-15,-3){\line(1,0){15}}\put(-18,-3){$j$}
\put(0,3){\line(2,-1){4}}
\put(0,-3){\line(2,1){4}}
\put(4,1){\line(1,0){4}}
\put(4,-1){\line(1,0){4}}
\put(8,-2){\framebox(4,4){$V$}}
\put(12,1){\line(1,0){3}}
\put(12,-1){\line(1,0){3}}

\put(-15,-7.5){\line(1,0){30}}
\put(-7,4){$\vdots$}
\put(-7,-1){$\vdots$}
\put(-7,-6.5){$\vdots$}
\end{picture}
\end{center}
\caption{}\label{fig:4.3}
\end{figure}
For achieving the tomography, we shall use only one and two
mode gates.

\par To begin with we consider the two $1$-mode Weyl
displacement operators $W(2^{-\frac{1}{2}}\imath)$ and
$W(2^{-\frac{1}{2}})$ where $\imath = \sqrt{-1}$. Put 
\begin{eqnarray}
G_p & =& W(2^{-\frac{1}{2}}\imath) \label{eq:4.3}\\
G_q &=& W(2^{-\frac{1}{2}}) \label{eq:4.4}
\end{eqnarray} 
If $\mb{e}_j=(0,\cdots,0,1,0\cdots,0)^T$ with $1$ in the
$j$-th position then applying $G_p$ and $G_q$ in the $j$-th
mode is equivalent to using the displacement operator
$W(2^{-\frac{1}{2}}\imath \mb{e}_j)$ and
$W(2^{-\frac{1}{2}}\mb{e}_j)$ respectively. Then equation
(\ref{eq:4.1}) in Theorem \ref{th:4.1} reduces to
\begin{eqnarray}
l_j &=& \left \langle  {G_p^{(j)}}^\dag N G_p^{(j)}
\right\rangle -\langle N \rangle -1, \label{eq:4.5}\\
m_j &=& \left \langle  {G_q^{(j)}}^\dag N G_q^{(j)}
\right\rangle -\langle N \rangle -1. \label{eq:4.6}
\end{eqnarray}
Furthermore (\ref{eq:3.4}) implies 
\begin{equation}\label{eq:4.7}
\tr~S = 2\langle N \rangle - \|\mb{l}\|^2 - \|\mb{m}\|^2 +
n.
\end{equation}
In other words, the measurement of counting observables $N,~
{G_p^{(j)}}^\dag N G_p^{(j)}, ~ {G_q^{(j)}}^\dag N
G_q^{(j)}, ~ 1 \le j \le n$ which constitute a set of
cardinality $2n+1$ yields the $2n+1$ parameters $l_j,~ m_j~
(1 \le j \le n)$ and $\tr~S$ concerning the Gaussian state
$\rho_g(\mb{l}, \mb{m}; S)$. 

\par Till now, the only resources we have used are identity
(or zero mode) gate and the one mode gates $G_p$ and $G_q$.
Now we shall pass on to $1$-mode Gaussian symmetry gates.

\par We consider the element $L(x,\alpha) \in Sp(2)$ defined
by
\begin{equation}\label{eq:4.8}
L(x,\alpha) = \begin{bmatrix} \cos \alpha & -\sin \alpha \\
\sin \alpha & \cos \alpha \end{bmatrix} 
\begin{bmatrix} x & 0 \\ 0 & \frac{1}{x} \end{bmatrix}
\begin{bmatrix} \cos \alpha & \sin \alpha \\
-\sin \alpha & \cos \alpha \end{bmatrix} , \quad x>1,\quad 0
\le \alpha < 2\pi
\end{equation}
and the unitary operator 
\begin{equation}\label{eq:4.9}
G_{sp}(x, \alpha) = \Gamma(\tau(L(x, \alpha)))
\end{equation}
where $\tau(L)=(L^{-1})^T$ in any symplectic group. We now
view $G_{sp}(x, \alpha)$  as a $1$-mode gate and apply it in
different modes. We express it by Figure \ref{fig:4.4}
\begin{figure}[h]
\begin{center}
\setlength{\unitlength}{2mm}
\begin{picture}(40,10)(-30,-5)
\put(-15,4.5){\line(1,0){25}}
\put(-15,3){\line(1,0){25}}
\put(-15,0){\line(1,0){11}}
\put(-4,-2){\framebox(10,4){$ G_{sp}(x,\alpha)$}}
\put(6,0){\line(1,0){4}}
\put(-15,-4.5){\line(1,0){25}}
\put(-10,0.5){$\vdots$}
\put(-10,-3){$\vdots$}
\put(-30,0){$G_{sp}^j(x,\alpha)=$}
\end{picture}
\end{center}
\caption{}\label{fig:4.4} 
\end{figure}
where the box is in the $j$-th wire. Expressing the
covariance matrix $S$ of the $n$-mode Gaussian state as a
block matrix
\[S =[[S_{i,j}]], \qquad i,j \in \{1,2,\cdots,n\}\]
where each $S_{i,j}$ is a $2 \times 2$ matrix we observe
that $S_{jj}$ is the covariance matrix of the $j$-th marginal
Gaussian state and
\begin{equation}\label{eq:4.10}
\begin{bmatrix}
S_{ii} & S_{ij}\\
S_{ij}^T & S_{jj} 
\end{bmatrix}, \quad i<j
\end{equation}
is the covariance matrix of the $i~j$-marginal Gaussian
state of $\rho_g(\mb{l}, \mb{m};S)$.

\par Let 
\begin{equation}\label{eq:4.11}
S_{jj}=\begin{bmatrix}
\sigma_{pp} & \sigma_{pq} \\
\sigma_{qp} & \sigma_{qq}
\end{bmatrix}
\end{equation}
with $\sigma_{qp} =\sigma_{pq}$. Thus $S_{jj}$ has three
parameters. Applying the gate $G_{sp}^{(j)}(x, \alpha)$
defined by (\ref{eq:4.9}) and Figure \ref{fig:4.4}, and
using part (ii) of Theorem \ref{th:4.1} we get
\begin{multline}\label{eq:4.12}
(x^2\cos^2 \alpha + x^{-2} \sin^2 \alpha -1) \sigma_{pp} +
(x^2\sin^2 \alpha + x^{-2} \cos^2 \alpha -1) \sigma_{qq} +
2(x^2 - x^{-2})\sin \alpha \cos \alpha \\
= 2 \left( \left\langle G_{sp}^{(j)}(x, \alpha)^\dag N
G_{sp}^{(j)}(x,\alpha) \right\rangle  -\langle N \rangle \right) +
[1-(x^2\cos^2 \alpha +x^{-2} \sin^2 \alpha )]l_j^2\\ +
[1-(x^2\sin^2 \alpha +x^{-2} \cos^2 \alpha )]m_j^2 + 2l_j
m_j (x^2 - x^{-2}) \sin \alpha \cos\alpha
\end{multline}
Choosing $\alpha=0$, (\ref{eq:4.12}) becomes
\begin{eqnarray}\label{eq:4.13}
&&(x^2 -1)\sigma_{pp} + (x^{-2} -1)\sigma_{qq} \nonumber\\
&=& 2 \left(\left\langle G_{sp}^{(j)}(x, 0)^\dag N
G_{sp}^{(j)}(x,
0) \right\rangle  -\langle N \rangle \right) + (1-x^2) l_j^2
+(1-x^{-2})m_j^2.
\end{eqnarray}
Choosing $x=\sqrt{2}$ and $x=\sqrt{3}$ successively
(\ref{eq:4.13}) yields two linearly independent equations
for determining $\sigma_{pp}$ and $\sigma_{qq}$ in terms of
$l_j,~ m_j,~ \langle N \rangle$ and $\left\langle
G_{sp}^{(j)}(x, 0)^\dag N G_{sp}^{(j)}(x,
0) \right\rangle$.
\par Now we go back to the equation (\ref{eq:4.12}) and
choose $x=\sqrt{2}$, $\alpha =\frac{\pi}{4}$. Then we get
\begin{eqnarray}
\frac{1}{4} (\sigma_{pp}+\sigma_{qq}) +\frac{3}{2}
\sigma_{pq} &=& 2\left(\left\langle G_{sp}^{(j)}\left(\sqrt{2},
\frac{\pi}{4}\right)^\dag N G_{sp}^{(j)}\left(\sqrt{2},
\frac{\pi}{4}\right)
\right\rangle - \langle N \rangle \right) \nonumber\\
&& -\frac{1}{4} (l_j^2+m_j^2) +\frac{3}{2} l_j m_j.
\label{eq:4.14}
\end{eqnarray}
Since $\sigma_{pp},~\sigma_{qq}$ and $l_j,~m_j$ have already
been determined, (\ref{eq:4.14}) determines $\sigma_{pq}$ by
using values of $\langle  N \rangle $ and $\left\langle
G_{sp}^{(j)}\left(\sqrt{2}, \frac{\pi}{4}\right)^\dag N
G_{sp}^{(j)}\left(\sqrt{2}, \frac{\pi}{4}\right)
\right\rangle$. Thus $S_{jj}$  can be completely determined
by measurements of $N$ using the gates $G_{sp}(\sqrt{2},0),~
G_{sp}(\sqrt{3},0)$ and $G_{sp}(\sqrt{2},\frac{\pi}{4})$ in
different modes after knowing the vectors $\mb{l}$ and
$\mb{m}$. However, after determining $S_{jj}$ for $1 \le j
\le n-1$, in order to determine $S_{nn}$ it is enough to use
only the two gates $G_{sp}(\sqrt{2},0)$ and
$G_{sp}(\sqrt{2},\frac{\pi}{4})$ in the $n$-th mode because
we already know $\tr\,S$ form (\ref{eq:4.7}). Thus we need
only $(3n-1)$ new measurements to determine the $3n$
parameters occurring in the block diagonals $S_{jj},~1\le j
\le n$.

\par Now it remains to determine for any $i<j$ the
off-diagonal block $S_{ij}$. To achieve this goal we shall
use a $2$-mode gate of the form $\Gamma(L),~L\in Sp(4)$. We
start with a unitary matrix of order $2$ of the form 
\[U = \begin{pmatrix} \alpha & \beta \\ -\bar{\beta} &
\bar{\alpha} \end{pmatrix}, \quad |\alpha|^2 + |\beta|^2 =1\]
where $\alpha = \alpha_1 + \imath\,\alpha_2,~
\beta=\beta_1+\imath\, \beta_2$ with $\alpha_j,~\beta_j$
being real. If we view  $U$ as a real linear transformation
of $\mathbb{R}^4$ we get an element of $Sp(4)$ of the form 
\begin{equation}\label{eq:4.15}
\mathcal{O} = \left[\begin{array}{rrrr}
\alpha_1 & -\alpha_2 & \beta_1 & -\beta_2 \\
\alpha_2 & \alpha_1 & \beta_2 & \beta_1 \\
-\beta_1 & -\beta_2 & \alpha_1 & \alpha_2\\
\beta_2 & -\beta_1 & -\alpha_2 & \alpha_1
\end{array}\right]
\end{equation}
where $\mathcal{O}$ is a real orthogonal matrix. Define
\begin{equation}\label{eq:4.16}
L(U,x_1,x_2) =\mathcal{O} 
\begin{bmatrix}
x_1 &&&\\
& x_1^{-1}&&\\
&& x_2&\\
&&&x_2^{-1}
\end{bmatrix} \mathcal{O}^T, \quad x_1>1,~x_2>1.
\end{equation}
We write 
\begin{equation}\label{eq:4.17}
L(U,x_1,x_2)^T L(U,x_1,x_2) =\begin{bmatrix}
A & B^T \\ B & C \end{bmatrix}
\end{equation}
and note that $A,~C$ are $2 \times 2$ positive definite
matrices and $B$ is given by 
\begin{equation}\label{eq:4.18}
B = \begin{bmatrix}
-\beta_1 \alpha_1 (x_1 -x_2) + \beta_2 \alpha_2 (x_1^{-1}
-x_2^{-1}) & -\beta_1 \alpha_2 (x_1 -x_2^{-1}) - \beta_2
\alpha_1 (x_1^{-1} -x_2)\\
\beta_2 \alpha_1 (x_1 -x_2^{-1}) + \beta_1 \alpha_2
(x_1^{-1} -x_2) & \beta_2 \alpha_2 (x_1 -x_2) - \beta_1
\alpha_1 (x_1^{-1} -x_2^{-1})
\end{bmatrix}.
\end{equation}

\par Using (\ref{eq:4.10}) and (\ref{eq:4.17}) we observe
that 
\begin{eqnarray}\label{eq:4.19}
\tr~\begin{bmatrix} S_{ii} & S_{ij} \\ S_{ji} & S_{jj}
\end{bmatrix}  L(U,x_1,x_2)^T L(U,x_1,x_2) &=& 
\tr~ \begin{bmatrix} S_{ii} & S_{ij} \\ S_{ji} & S_{jj}
\end{bmatrix}  \begin{bmatrix}
A & B^T \\ B & C \end{bmatrix}\nonumber \\
&=& \tr(S_{ii}A+S_{jj}C)+2\tr\,S_{ij}B
\end{eqnarray}
where the first sum on the right hand side depends only on
$S_{ii}$ and $S_{jj}$ which have already been determined in
terms of the expectations of $N$ and its conjugates by chosen
one mode gates. In order to determine $S_{ij}$ we use the
$2$-mode gate
\begin{equation}\label{eq:4.20}
G_{sp}(U,x_1,x_2) = \Gamma(\tau(L(U,x_1,x_2)))
\end{equation}
in the $(i,j)$-modes and use part (ii) of Theorem
\ref{th:4.1} to obtain the relation

\begin{eqnarray*}
&& \left\langle G_{sp}^{i,j}(U,x_1,x_2)^\dag  N
G_{sp}^{i,j}(U,x_1,x_2) \right\rangle -\langle N \rangle \\
&=&\frac{1}{2} \left[ \tr\begin{bmatrix} S_{ii} & S_{ij} \\ S_{ji} & S_{jj}
\end{bmatrix}  \begin{bmatrix}
A - I_2 & B^T \\ B & C-I_2\end{bmatrix} + \begin{pmatrix}
l_1 \\ -m_1  \\  l_2 \\ -m_2
\end{pmatrix}^T   \begin{bmatrix}
A - I_2 & B^T \\ B & C-I_2\end{bmatrix}\begin{pmatrix}
l_1 \\ -m_1  \\  l_2 \\ -m_2
\end{pmatrix}\right]. 
\end{eqnarray*}
Using (\ref{eq:4.19}) this reduces to 
\begin{eqnarray}
\tr~S_{ij}B &=& \left\langle G_{sp}^{i,j}(U,x_1,x_2)^\dag N
G_{sp}^{i,j}(U,x_1,x_2)\right\rangle -\langle N \rangle
\nonumber\\
&&  -\frac{1}{2} \tr~\left[ S_{ii}(A-I_2)
+S_{jj}(C-I_2)\right] \nonumber\\
&=& f(U,x_1,x_2) ~\text{say.}\label{eq:4.21}
\end{eqnarray}
When $l_i,~m_i,~l_j,~m_j,~S_{ii}$ and $S_{jj}$ are already
determined, the term $f(U,x_1,x_2)$ depends only on $U,~x_1,~x_2$.
Let 
\[S_{ij} =\begin{bmatrix} \gamma_{11} & \gamma_{12} \\
\gamma_{21} & \gamma_{22} \end{bmatrix}.\]
We now make four special cases for $(U,x_1,x_2)$.
\begin{enumerate}[(i)]
\item $U=H=\frac{1}{\sqrt{2}} \begin{bmatrix} 1 & 1 \\ -1 &
1 \end{bmatrix}, ~x_1=1,~x_2=2$. Put $r_1 = f(H,1,2)$. Then
(\ref{eq:4.21}) becomes
\begin{equation}\label{eq:4.22}
\frac{1}{2} \gamma_{11} - \frac{1}{4} \gamma_{22} =r_1.
\end{equation}
\item $U=H, ~x_1=1,~x_2=3$. Put $r_2 = f(H,1,3)$. Then
(\ref{eq:4.21}) becomes
\begin{equation}\label{eq:4.23}
 \gamma_{11} - \frac{1}{3} \gamma_{22} =r_2.
\end{equation}
\item $U=K=\frac{1}{\sqrt{2}} \begin{bmatrix} \imath & 1 \\ 
-1 & -\imath \end{bmatrix}, ~x_1=1,~x_2=2$. Put $r_3 = f(K,1,2)$. Then
(\ref{eq:4.21}) becomes
\begin{equation}\label{eq:4.24}
-\frac{1}{4} (\gamma_{21} +2  \gamma_{12}) =r_3.
\end{equation}
\item $U=K, ~x_1=1,~x_2=3$. Put $r_4 = f(K,1,3)$. Then
(\ref{eq:4.21}) becomes
\begin{equation}\label{eq:4.25}
-\frac{1}{3}( \gamma_{21} +3  \gamma_{12}) =r_4.
\end{equation}
\end{enumerate}
The four equations (\ref{eq:4.22})--(\ref{eq:4.25}) in the
unknowns $\gamma_{11},~\gamma_{22},~\gamma_{12},~\gamma_{21}$
are linear and linearly independent. Thus they determine the
matrix $S_{ij}$ for any fixed $i,~j$. For this purpose we
have used exactly four measurements described by the four
$2$-mode gates for four parameters. 

\par In all we have used exactly $(2n+1) + (3n-1) +4
\frac{n(n-1)}{2} =n(2n+3)$ measurements to determine the
$n(2n+3)$ parameters of the Gaussian state $\rho_g(\mb{l},
\mb{m}; S)$.


\section{Tomography of Gaussian channels}\label{sec:channel}
\par An $n$-mode Gaussian channel $\mathcal{K}(A,B)$ is
described by a pair of real $2n \times 2n$ matrices $(A,B)$
where $B$ is positive semidefinite and the following matrix
inequality holds:
\[B +\imath(A^TJ_{2n}A-J_{2n})\ge 0\]
with $J_{2n}$ as in equation (\ref{eq:2.12}). Thus
$\mathcal{K}(A,B)$ is determined by $6n^2+n$ real
parameters. Such a channel has the property that for any
Gaussian input state $\rho_g( \mb{l},\mb{m};S)$, the
corresponding output is again Gaussian and has the form
$\rho_g( \mb{l}',\mb{m}';S')$ where 
\begin{eqnarray}
\begin{pmatrix} \mb{l}' \\ -\mb{m}' \end{pmatrix} &=& 
A^T \begin{pmatrix} \mb{l} \\ -\mb{m} \end{pmatrix}
\label{eq:5.1}\\
S' &=& A^TSA +\frac{1}{2} B.\label{eq:5.2}
\end{eqnarray}
our aim is to determine $A$ and $B$ by performing tomography
on the output states for a small number of coherent input
statea $\ket{\psi(\mb{u})}$ where
\[ |\psi(\mb{u})\rangle \langle \psi(\mb{u})|=
\rho_g\left(\sqrt{2} \mb{y}, \sqrt{2}\mb{x};
\frac{1}{2}I_{2n}\right)\]
where $\mb{x}=\re(\mb{u}), ~\mb{y}=\im(\mb{u})$. By
(\ref{eq:5.1}) and (\ref{eq:5.2}) the output state is 
\begin{equation}\label{eq:5.3}
\rho_g\left(\mb{y'},\mb{x'}; \frac{1}{2}(A^TA+B)\right)
\end{equation}
where 
\begin{equation}\label{eq:5.4}
\begin{pmatrix} \mb{y'} \\ -\mb{x'} \end{pmatrix} =
A^T \begin{pmatrix} \sqrt{2} \mb{y} \\ - \sqrt{2}\mb{x}
\end{pmatrix}.
\end{equation}
\par Now we specialize the values of $\mb{u}$ and select the
$2n$ input states 
\begin{eqnarray}
\rho_j &=& |\psi(2^{-\frac{1}{2}}\imath \mb{e_j}) \rangle
\langle \psi(2^{-\frac{1}{2}}\imath \mb{e_j})|, \quad 1 \le j
\le n, \label{eq:5.5}\\
 \rho_j' &=& |\psi(2^{-\frac{1}{2}} \mb{e_j} )\rangle
\langle \psi(2^{-\frac{1}{2}} \mb{e_j})|, \quad 1 \le j
\le n. \label{eq:5.6}
\end{eqnarray}
Then the corresponding output states are 
\begin{equation}\label{eq:5.7}
\tilde{\rho_g}\left(\tilde{\mb{l_j}}, \tilde{\mb{m_j}},
\frac{1}{2} (A^TA+B)\right)
\end{equation}
with 
\begin{equation}\label{eq:5.8}
\begin{pmatrix} \tilde{\mb{l_j}}\\ -\tilde{\mb{m_j}}
\end{pmatrix} = A^T \begin{pmatrix} \mb{e_j}\\\mb{0}
\end{pmatrix}, \quad 1 \le j \le n 
\end{equation}
and 
\begin{equation}\label{eq:5.9}
\tilde{\rho_g}'\left(\tilde{\mb{l_j}}', \tilde{\mb{m_j}}',
\frac{1}{2} (A^TA+B)\right)
\end{equation}
with 
\begin{equation}\label{eq:5.10}
\begin{pmatrix} \tilde{\mb{l_j}}'\\ -\tilde{\mb{m_j}}'
\end{pmatrix} = A^T \begin{pmatrix} \mb{0}\\ \mb{e_j}
\end{pmatrix}, \quad 1 \le j \le n .
\end{equation}

\par A full tomography on (\ref{eq:5.7}) with $j=1$ as
outlined in Section \S\ref{sec:state} yields the first row
of $A$ and the matrix $A^TA+B$ by using $n(2n+3)$
measurements. A similar but partial tomography of the
remaining $(n-1)$ states in (\ref{eq:5.7}) and all the
states in (\ref{eq:5.9}) but only for the mean values yields the
remaining $(2n-1)$ rows of $A$. This needs an additional set
of $(2n-1)(2n+1) =4n^2 -1$ measurements. In all, our
approach requires $6n^2+3n-1$ measurements for getting the
$6n^2+n$ parameters. It will be interesting to know whether
one can determine $A$ and $B$ with less measurements. If
this is not possible our problem will carry an intrinsic
tomographic complexity.


\section{Conclusions}

\par All the $2n^2+3n$ mean and covariance parameters of an
$n$-mode Gaussian state can be recovered from the
expectation values of the same number of conjugates of the
total number operator by Gaussian symmetries. Such
symmetries can be realised by five one mode and four two
mode gates.  The complete tomography of a Gaussian state can
be expressed by circuit diagrams and measurements akin to
those in quantum computation theory. An application of this
tomography to the output of an $n$-mode Gaussian channel
corresponding to appropriate coherent inputs determines all
the $6n^2+n$ parameters by $6n^2+3n-1$ measurements.
Improvement in this channel tomography and finding the
probability distribution of the number operator from its
explicitly computable probability generating function in a
general $n$-mode Gaussian state seem to be interesting
problems arising from our investigations.

\bibliographystyle{alpha}
\bibliography{biblio}
\end{document}